\RequirePackage{amsmath}
\documentclass{llncs}
\usepackage[utf8]{inputenc}
\usepackage{amssymb}
\usepackage[usenames]{xcolor}
\usepackage{listings}
\usepackage{inconsolata}
\usepackage{wasysym}
\usepackage{booktabs}
\usepackage{stmaryrd}
\usepackage{scalerel}
\usepackage{framed}
\usepackage{tikz}
\usepackage[draft]{todonotes}
\usetikzlibrary{positioning}

\newcommand{\answer}{\operatorname{ans}}

\newcommand{\filter}{\operatorname{filter}}
\newcommand{\define}{\operatorname{let}}

\newcommand{\infer}{\operatorname{fact}}

\colorlet{punct}{red!60!black}
\definecolor{background}{HTML}{EEEEEE}
\definecolor{delim}{RGB}{120,20,40}
\colorlet{numb}{magenta!60!black}

\lstdefinelanguage{sparql}{
	sensitive=false,
	extendedchars=true,
	literate={á}{{\'a}}1 {é}{{\'e}}1 {í}{{\'{\i}}}1 {ó}{{\'o}}1 {ú}{{\'u}}1
	{Á}{{\'A}}1 {É}{{\'E}}1 {Í}{{\'I}}1 {Ó}{{\'O}}1 {Ú}{{\'U}}1
	{ü}{{\"u}}1 {Ü}{{\"U}}1 {ñ}{{\~n}}1 {Ñ}{{\~N}}1 {¿}{{?``}}1 {¡}{{!``}}1
	{<}{{{\color{delim}<}}}{1}
        {>}{{{\color{delim}>}}}{1}
	{?}{{{\color{delim}?}}}{1}
	{*}{{{\color{delim}*}}}{1}
	{+}{{{\color{delim}+}}}{1}
	{/}{{{\color{delim}/}}}{1}
	{,}{{{\color{punct}{,}}}}{1}
        {;}{{{\color{punct}{;}}}}{1}
        {.}{{{\color{punct}{.}}}}{1}
        {:}{{{\color{punct}{:}}}}{1},
	morekeywords={select,from,where,order,by,distinct,limit,offset,optional,union,filter,prefix,bound,desc,regex,str,group,not,exists,minus,service,certain,maybe,and,as}
}

\lstdefinestyle{sparqld}{
	basicstyle=\small\ttfamily,
	identifierstyle=\color{black},
	keywordstyle=\color{blue!40!black}\bfseries,
	ndkeywordstyle=\color{greenCode}\bfseries,
	stringstyle=\color{ocherCode}\ttfamily,
	commentstyle=\color{darkgray}\ttfamily,
	language={sparql},
	tabsize=2,
	showtabs=false,
	showspaces=false,
	showstringspaces=false,
	extendedchars=true,
	escapechar=`,
	frame={single},
	breaklines=true,
	basewidth=0.5em,
	moredelim=[is][\color{magenta}]{~}{~},
	moredelim=**[is][\color{gray}]{£}{£},
	moredelim=**[is][\color{blue!50!black}]{$}{$},
	moredelim=[is][\color{orange!80!black}]{!}{!},
	moredelim=**[is][\color{green!50!black}]{¬}{¬},
    xleftmargin=2ex,
    xrightmargin=1ex,
    aboveskip=1.5ex,
    belowskip=1.5ex
}

\DeclareFontEncoding{LS1}{}{}
\DeclareFontSubstitution{LS1}{stix}{m}{n}
\DeclareSymbolFont{symbols2}{LS1}{stixfrak}{m}{n}
\DeclareMathSymbol{\leftouterjoin}{\mathbin}{symbols2}{"11}
\newcommand{\LeftOuterJoin}{\mathrel{\scaleobj{0.8}{\leftouterjoin}}}

\begin{document}

\title{The Problem of Correlation and \\Substitution in SPARQL\\
{\normalsize\sc ---Extended Version---}}
\titlerunning{Correlation and Substitution in SPARQL}

\author{%
  Daniel Hernández\inst{1} \and
  Claudio Gutierrez\inst{1} \and
  Renzo Angles\inst{2}}

\authorrunning{Daniel Hernández et al.} 

\tocauthor{Daniel Hernández, Claudio Gutierrez, Renzo Angles}

\institute{%
  Universidad de Chile, Santiago, Chile \and
  Universidad de Talca, Curicó, Chile}

\maketitle              

\begin{abstract}
  Implementations of a standard language are expected to give same
  outputs to identical queries.  In this paper we study why different
  implementations of SPARQL (Fuseki, Virtuoso, Blazegraph and rdf4j)
  behave differently when evaluating queries with correlated
  variables.  We show that at the core of this problem lies the
  historically troubling notion of logical substitution.  We present a
  formal framework to study this issue based on Datalog that besides
  clarifying the problem, gives a solid base to define and implement
  nesting.  
\end{abstract}

\section{Introduction}

A {\em subquery} is a query expression that occurs in the body of
another query expression, called the {\em outer query}.  A {\em
  correlated subquery} is one whose evaluation is dependent in some
way on data being processed in the outer query. Informally the data
got from the outer query should be replaced or {\em substituted} in
the corresponding places in the inner query.  Thus the notion of {\em
  substitution} comes to the heart of the problem at hand.

It is well known the complexities that this notion involves.
As is well known, it has been always a troubling concept and source of
error even to renowned logicians.\footnote{The story is recounted by
  Church in \cite{church1996introduction}, pp. 289-90 and Cardone and
  Hindley \cite{cardone2006history}, p.7. Russell and Whitehead,
  although used the notion, missed its formal statement in his {\em
    Principia} (1913). Hilbert and Ackermann gave an ``inadequate''
  statement in his 1928's Logic. Carnap in {\em Logicshe Syntax der
    Sprache} and Quine in {\em System of Logistic} definitions still
  contain problems.  Hilbert and Ackermann in 1934 gave finally a
  correct statement.}  The query language SPARQL is not an exception:
the notion of replacement (substitution) in the recommendation has an
insufficient definition and even contradictory pieces~\footnote{See
  \cite{DBLP:journals/corr/HernandezGA16} where we report that the
  substitution notion presented on Sec. 18.6 of the current
  specification \cite{sparql11} contradicts the statements ``Due to
  the bottom-up nature of SPARQL query evaluation, the subqueries are
  evaluated logically first, and the results are projected up to the
  outer query'' and ``Note that only variables projected out of the
  subquery will be visible, or in scope, to the outer query'', from
  Sec. 12.}.

We show in this paper that this notion is the source of problems that 
can be found in the evaluation of EXIST subqueries in SPARQL. The problem 
is highlighted when operators which incorporate possibly incomplete
information are present. 
The  following SPARQL query illustrates these problems. The query 
roughly asks for the id of persons and optionally 
their corporate email, subject to some conditions given by the expression in the
FILTER EXISTS:

\begin{lstlisting}[style=sparqld]
SELECT ?id ?email
WHERE { {{ ?id a :person } OPTIONAL { ?id :corpMail ?email }}
        FILTER EXISTS {
            {{ ?id a :person } OPTIONAL { ?id :privMail ?email }}
            FILTER (?email = *.com) } }
\end{lstlisting}
Surely the reader is facing the following problem: how to
interpret this query?  Well, you are not alone in your
vacillation: The most popular implementations of SPARQL do not
agree on it. For example, for the following database of persons
\setlength\tabcolsep{1em}
\begin{center}
\begin{tabular}{|c|c|c|c|c|c|c|}
\hline
id       &\tt    1&\tt    2&\tt    3&\tt    4&\tt    5&\tt    6\\\hline
corpMail &\tt*.com&\tt*.net&\tt*.com&\tt*.net&        &\tt     \\\hline
privMail &\tt*.net&\tt*.com&        &        &\tt*.com&\tt*.net\\\hline
\end{tabular}
\end{center}
Fuseki, Blazegraph, Virtuoso and rdf4j give almost all different
results.  Fuseki and Blazegraph give (\texttt{1},\texttt{*.com}),
(\texttt{3},\texttt{*.com}) and (\texttt{5},-); Virtuoso gives
(\texttt{1},\texttt{*.com}) and (\texttt{3},\texttt{*.com}); finally
rdf4j gives (\texttt{3},\texttt{*.com}) and
(\texttt{5},-).\footnote{The engines studied in this paper are Fuseki
  2.5.0, Blazegraph 2.1.1, Virtuoso 7.2.4.2, and rdf4j 2.2.1.}

What is going on? Not simple to unravel. The main problem is
how to assign the variables in order to evaluate the inner 
and outer expressions. Let us see
why not all systems agree in showing up person \verb|1|.  If we
evaluate first the inner query then the variable \verb|?email| is
bound to \verb|*.net| for person \verb|1|, thus the filter
\verb|?email| = \verb|*.com| fails, and so the whole expression inside the 
first filter fails, hence person \verb|1| is not shown in the output.
Now, if the inner pattern is evaluated {\em after} bounding the 
variable \verb|?email| to \verb|*.com|, then the OPTIONAL part does not
match and thus the last filter pass, hence person \verb|1| is outputted.

The intuition provided by the case of person \verb|1| is that some
systems evaluate the inner pattern before binding the \verb|?email| in the
outer query, and others do it after the binding.
 Now this intuitive philosophy does not work
to understand why Blazegraph and Fuseki outputs,
(\texttt{1},\texttt{*.com}) and (\texttt{5},-).
 In defense of these systems, let us recall that 
the specification is not clear or precise enough about this evaluation.



We ---the Semantic Web advocates--- are in a problem:
No query language with such uncertainties will gain wide adoption. Of
course the example is highly unnatural but, as we will show, it codes
the essence of the problems of substitution of correlated variables in
SPARQL.  Substituting is not trivial even when all values are
constants, but when considering incomplete information, e.g.  in the
form of nulls or unbounds in SPARQL, the complexities rapidly scale
up.  In this setting, having a clean logical picture of what is going
on is crucial.

Our aim is that the above problems can be modeled by using Datalog.
Hence, we introduce Nested Datalog, an extension of classical Datalog
to cope with nested expressions.  Nested Datalog is used to describe
two different substitution philosophies: {\em syntactic substitution},
which implies that the value to be substituted in a correlated
variable comes from one source; and {\em logical} substitution, which
implies that a variable is expecting values from two sources (from the
valuation of the expression in the database and from the outer query).
  
  The problem becomes more complex in the logical 
  evaluation when one (or both) of the values is null.
  It turns out that it is not indifferent where we 
  ``compatibilize" both values (at the bottom of the 
  program, at the top, in the middle).
    To model this it is needed one step further in the
    formalization in order to capture the nuances of 
    incomplete information in the form of null values.
To do this, we use {\em Modal Datalog}, a version of Datalog with modal features.
 
Once having the right setting, the subtleties of the notion of
substitution in nested expressions having incomplete information (null
values) becomes visible. Then, in Section 5, we show how SPARQL
translates to this setting, and how the notion of substitution is
expressed in it. Then we show the discrepancies of the different
implementations and what they mean under the light of this formal
framework.  We conclude with our view of how to handle this problem.

In summary, the contributions of this paper are the following:

\begin{enumerate}
\item Formalization of the problem (i.e. substitution in nested
  expressions), its study from a logical point of view, and analysis
  of the discrepancies of evaluation under different implementations
  of SPARQL.

\item To provide a logical framework to understand and formalize the
  notion of substitution in nested expressions in the presence of
  incomplete information for SPARQL.
   
\item Presentation of the logical (and consistent) alternatives
  defined and supported by our logical framework.
\end{enumerate}

\paragraph{Related work.}
Different problems related to the notion of substitution are listed in
the SPARQL specification errata. In a previous
work~\cite{DBLP:journals/corr/HernandezGA16} we have reported some
issues and presented three alternative solutions based on rewriting on
the nested query before the substitution. After that, a W3C community
group\footnote{https://www.w3.org/community/sparql-exists/} was
created to address these issues.  The community started defining
queries and their expected outputs, for two alternative semantics.
The first, proposed by Patel-Schneider and
Martin~\cite{DBLP:conf/semweb/Patel-Schneider16}, and the second
proposed by Seaborne in the mailing list of the community group. None
of these proposals study the problem in a formal framework as we do in
this paper.

The idea of nested queries in Datalog is not new (see
\cite{DBLP:journals/jlp/GiordanoM94} for example).  We introduce
nesting in a different way in order to be more suitable for studying
correlation.  Similarly, null values have been already studied in
deductive databases
(e.g.,~\cite{DBLP:conf/iclp/DongL92,DBLP:journals/tcs/DongL94,DBLP:journals/tois/KongC95}),
but with a focus (on computing certain answers) that is not the goal
of this paper.


\section{Standard Datalog}
\label{sec:standard-datalog}

We will briefly review notions of non recursive Datalog with
equalities and safe negation. For further details
see~\cite{DBLP:books/daglib/0097689}.

\paragraph{Datalog Syntax.}
A \textit{term} is either a variable or a constant.  An \textit{atom}
is either a \textit{predicate formula} $p(t_1,...,t_n)$ where $p$ is a
predicate name and each $t_i$ is a term, or an \textit{equality
  formula} $\filter(t_1 = t_2)$ where $t_1$ and $t_2$ are terms.  A
\textit{literal} is either an atom (a \textit{positive literal} $L$)
or the negation of an atom (a \textit{negative literal} $\neg L$).
For readability we also write $\filter(t_1\neq t_2)$ instead of
$\neg \filter(t_1 = t_2)$.  A \textit{rule} $R$ is an expression of
the form $L \gets L_1, \dots, L_n$ where $L$ is a predicate formula
called the \textit{head} of $R$, and $L_1, \dots, L_n$ is a set of
literals called the \textit{body} of $R$.  A \textit{fact} $F$ is a
predicate formula with only constants. A \textit{program} $P$ is a
finite set of rules and facts.  A \textit{query} $Q$ is a pair $(L,P)$
where $L$ is a predicate formula called the \textit{goal} of $Q$, and
$P$ is a program.  We assume that all predicate formulas in a query
$Q$ with the same predicate name have the same number of arguments and
that the terms in the head of each rule in $Q$ and the goal of $Q$ are
all variables.

%
%
%

\paragraph{Semantics of Datalog.}
Given a substitution $\theta=\{X_1/a_,\dots,X_n/a_m\}$ from variables
to constants, and a literal $L$, $\theta(L)$ denotes the literal
resulting of substituting in $L$ each occurrence of $X_i$ by $a_i$,
for $1\leq i \leq m$.  Given a set of facts $S$, a substitution
$\theta$, and a positive literal $L$, we say that $S$ models $L$ with
respect to $\theta$, denoted $S,\theta\models L$, if either
$\theta(L)$ is an equality formula of the form $\filter(a=a)$ where
$a$ is a constant, or $L$ is predicate formula and $\theta(L)\in S$.
Similarly, given a set of facts $S$, a substitution $\theta$ and a
negative literal $\neg L$, $S,\theta\models \neg L$, if
$S,\theta\not\models L$.  Given a rule $R=L_0\leftarrow L_1,\dots,L_n$,
a fact $\theta(L_0)$ is inferred in a set of facts $S$ if
$S,\theta \models L_j$, for $1\leq j \leq n$.

A variable $X$ occurs positively in a rule $R$ if $X$
occurs in a positive predicate formula in the body of $R$.
A rule $R$ is said to be \textit{safe} if all its variables occur 
positively in $R$.  A program is \textit{safe} if all its rules
are safe.  The safety restriction provides a syntactic
restriction of programs which enforces the finiteness of derived
predicates.

In this paper we do not consider equality formulas $X=a$
as in most works (e.g., \cite{DBLP:books/daglib/0097689}),
but equality formulas of the form $\filter(X=a)$.
They differ in the form of evaluation. If a rule $R$
has the literal $X=a$ in the body, then $X$ is said to be
defined positively because the equality assigns the value $a$
to $X$. On the contrary, here equality formulas
require that all variables be assigned before being evaluated.
Formally, here equality formulas are built-in.

The \textit{dependency graph} of a Datalog program $P$ is a digraph
$(N,E)$ where the set of nodes $N$ is the set of predicates that occur
in the literals of $P$, and there is an arc $(p_1,p_2)$ in $E$ for
each rule in $P$ whose body contains predicate $p_1$ and whose
head contains predicate $p_2$.  A Datalog program is said to be
\textit{recursive} if its dependency graph is cyclic, otherwise it is
said to be \textit{non-recursive}. On this report we consider only
non-recursive programs.

Let $P$ be a program, $E$ be a subset of the Herbrand base of $P$, and
$\infer(P,E)$ denote the set of facts occurring in $E$ or $P$,
intuitively the facts inferred
in zero steps.  Then, the \textit{meaning} of $P$ is the result of
adding to $\infer(P,E)$ as many new facts as can be inferred from the
rules of $P$ in $\infer(P,E)$.  The inference process is applied
repeatedly until a fixpoint, denoted $\infer^*(P,E)$, is reached. The
\textit{answer} of a query $Q=(L,P)$ in an extensional database $E$,
denoted $\answer_E(Q)$, is the subset of $\infer^*(P,E)$ with facts
having the same predicate as $L$.

The fixpoint depends on the order used to evaluate rules.
Here we assume the order of Stratified Datalog, where for every
arc $(p_1,p_2)$ in the dependency graph of a program $P$, a rule $R_2$
with head $p_2$ is not used in the inference process until every
rule $R_1$ with head $p_1$ cannot be applied to infer another fact.
Without this order a negative predicate formula $\neg L$
can be wrongly evaluated as true if the evaluation is done before
a fact matching $L$ is inferred.

\section{Nested Datalog}
\label{sec:nested-datalog}

We will extend Datalog in order to be able to compose queries, by
introducing nested queries as a new type of atom that occurs as a
filter device (i.e., a built-in).


\begin{definition}[Syntax of Nested Datalog defined recursively]
1. A Datalog query is a Nested Datalog query.
2. A Nested Datalog query is a Datalog query where Nested Datalog
queries are allowed as atoms.
\end{definition}

The inference process of Nested Datalog differs from the standard one
by the addition of the semantics for evaluating Nested queries.
If query $Q$ is an atom in a rule, its evaluation
with respect to  the substitution $\theta$ is true if and only if 
$\theta(Q)$ has at least one answer, where $\theta(Q)$ denotes the query resulting of
``applying'' to $Q$ the substitution $\theta$. This is the key notion
we will study in what follows.

We will need the notion of assignment of a value to a variable.  
Consider a program $P$ with a single rule $R$ defined as
$p(X) \leftarrow q(X),\filter(X=Y)$. The assignment of the value $a$ to the
variable $Y$
can be done by adding a literal
$l(Y)$ to $R$, where $l$ is a fresh predicate, and the fact $l(a)$ to $P$. 
In the resulting program $Y$ can only
take the value $a$. In what follows we will use the notation
$\define(Y=a)$ as a syntactical sugar to denote the result of
assigning $a$ to $Y$ in $R$. 

Substitution $\theta(Q)$ for a nested query $Q$ turns out to be rather subtle.
There are two main approaches that we will call \textit{syntactical} and 
\textit{logical}.

\paragraph{Syntactical substitution.}
It works like standard
replacement of a variable in an open sentence in logic or a free
variable in a programming language. It occurs when a rule cannot be
evaluated without having the value of a ``free'' variable occurring in
it. Formally, a variable $X$ occurs free in a rule $R$ when it does
not occur positively in $R$ and it occurs in an equality formula
or in a nested query $Q$ where recursively $X$ occurs free in a rule
of $Q$.

\begin{definition}[syntactic substitution]
  Given a substitution $\theta$, a program $P$ and a variable $X$
  occurring in $\theta$, the syntactical substitution of $X$ in
  $P$ with respect to $\theta$ is done 
  by adding the literal $\define(X=\theta(X))$ to the body of each
  rule of $P$ where $X$ is free.
\end{definition}


\paragraph{Logical substitution.}  This is the problem of ``substituting''
$\theta(X)$ in a program $P$ that has no ``free'' $X$ (i.e. all their
rules $R$ are ``logically closed'').  Conceptually, in this case the semantics is
one such that after finding a solution of $R$, it checks its ``compatibility'' with
$\theta$.

The essential problem of logical substitution is ``when'' (at what
point in the evaluation) we will test this compatibility.  For
example, consider the program $P$ with the rules
$p(X) \leftarrow q(X)$ and $q(Y) \leftarrow r(Y)$.  The variable $Y$
in the second rule is logically connected with the variable $X$ in the
first rule. Thus, we can alternatively do the substitution in two
places:
\begin{center}
	\begin{tabular}{|c|c|}
		\toprule
		Top-down & Bottom up \\ \midrule
		{\small
			\(
			\begin{array}{rl}
			p(X) &\leftarrow q(X), \define(X=\theta(X))\\
			q(Y) &\leftarrow r(Y)
			\end{array}
			\)
		}
		&
		{\small
			\(
			\begin{array}{rl}
			p(X) &\leftarrow q(X)\\
			q(Y) &\leftarrow r(Y), \define(Y=\theta(X))
			\end{array}
			\)
		}\\
		\bottomrule
	\end{tabular}
\end{center}
This example illustrates two extremes, (1) Top down: evaluate $P$
first, then, proceed to check the compatibility of the solution with
$\theta$.  (2) Bottom up: Check compatibility of $\theta$ (with the
database) before starting the evaluation of $P$.  However, there are
also several valid substitutions in the middle.  In fact, we can start
with a top-down substitution and then move the literals
$\define(X = \theta(X))$ down in the dependency graph of the program.
The method used to move substitutions is equivalent to the standard
method used to move selections $\sigma_{X=a}$ in Relational Algebra
because of optimization concerns.  In the appendix we provide a detailed
description of the process of moving substitutions. The following
result follows from it.

\begin{lemma}
  \label{lemma:equivalence-top-bottom}
  Given a substitution $\theta$ and a Nested Datalog program $P$,
  moving down literals of the form $\define(Y=\theta(X))$
  in the dependency graph of $P$ does not change the semantics of $P$.
\end{lemma}

Hitherto, we have defined logical substitution for a rule but not for
a whole query. We now will define it using the top-down
approach.\footnote{Other approaches can be obtained by moving the
  substitution point from upper levels to lower levels in the
  dependency graph of the query.}

\begin{definition}[Top-down logical substitution]
  Given a substitution $\theta$, a query $Q=(p(X_1,\dots,X_n),P)$ and
  a variable $X$ occurring in the goal of $Q$ and $\theta$, the
  top-down logical substitution of $X$ in $Q$ is the query resulting
  of replacing the goal of $Q$ by $q(X_1,\dots,X_n)$, where $q$ is a
  fresh predicate, and adding the rule
  $q(X_1,\dots,X_n) \leftarrow p(X_1,\dots,X_n), \define(X=\theta(X))$
  to $P$.
\end{definition}

Logical and syntactical substitutions are not arbitrary. They are motivated
by the EXISTS operator of SPARQL as is shown in the following query.
\begin{lstlisting}[style=sparqld]
SELECT ?X
WHERE { ?X :hasMail ?Y
        FILTER EXISTS { SELECT ?X
                        WHERE { ?X :hasMail ?Z FILTER (?Y <> ?Z) } } }
\end{lstlisting}
Intuitively this query finds people \verb|?X| with multiple emails.
The variable \verb|?X| cannot be substituted by a constant in the
nested query, because that substitution breaks the syntax of the
SELECT clause (where only variables are allowed). Thus, \verb|?X| has
to be substituted logically. On the other hand, \verb|?Y| is a free
variable in the nested query, so \verb|?Y| has to be substituted
syntactically. We claim that these substitutions are better understood
if we rewrite this SPARQL query as the Nested Datalog query
$(p(X,Y),P_1)$ where $P_1$ is the following program:
\begin{align*}
P_1:\quad& p(X,Y) \leftarrow \operatorname{mail}(X,Y), (q(X),P_2),\\
P_2:\quad& q(X) \leftarrow \operatorname{mail}(X,Z), \filter(Z\neq Y).
\end{align*}
Hence, the result of applying a substitution $\theta$ in $P_2$ produces
a program with the rule
$q(X) \leftarrow \operatorname{mail}(X,Z), \filter(Z\neq Y),
\define(X=\theta(X)), \define(Y=\theta(Y))$.

%

There is a third form of substitution, namely \textit{improper}.
To gain some intuition, consider the following SPARQL query:
\begin{lstlisting}[style=sparqld]
SELECT * WHERE { ?X :r ?Y FILTER EXISTS { SELECT ?Z WHERE { ?X :s ?Z } } }
\end{lstlisting}
The inner pattern of this query can be modeled as a query $(p(Z),P)$
where $P$ contains the rule $p(Z) \leftarrow s(X,Z)$.  The variable
$X$ cannot be substituted logically, because it does not occur in the
goal nor in the head of the unique rule of $P$. Furthermore, $X$
cannot be substituted syntactically, because it is not free. However,
some systems assume that $X$ is correlated. The behavior of these
systems coincides with applying improper substitution.  Here,
substitution is done by replacing the unique rule of $P$ by the rules
$p(Z) \leftarrow u(X,Z)$ and
$u(X,Z) \leftarrow s(X,Z), \define(X=\theta(X))$, where $u$ is a fresh
predicate.

Essentially, an improper substitution is a logical substitution that
starts in some point of the nested query instead of the goal of the
query, so that there is a gap in the logical chain. After that point,
substitution is similar to logical substitution in the sense that it
moves the value assigned to $X$ from the head of a rule to the body,
and thus to other rules below in the dependency graph.

In the example, the substitution is improper, because the logical
chain start on the second rule so there is a gap between the goal of
the nested query and the point where the logical chain starts.

Improper substitution breaks the design of Datalog where the scope of
variables is the rule where they occur. It disallows renaming
variables because its scope could be extended beyond the nested query,
so breaking the design of logic and compositional languages where non
free variables are renamed without changing the semantics of the
expression, and free variables cannot be renamed.  In our opinion,
this compromises the compositionality of queries.

Now we have a parametrical definition of $\theta(Q)$ for a query $Q$
(i.e., without the values of a particular substitution $\theta$),
independently of the kind of substitution used
(syntactical, logical or improper)
substitutions are used. A consequence, is the following result.

\begin{lemma}
  Given a Nested Datalog query $Q=(p(X_1,\dots,X_n),P)$ and a simple
  extensional database $E$ for $Q$, there exists a first order formula
  $\phi$ such that each
  $\answer_E(Q)=\{p(t_1,\dots,t_n) \mid E,\phi \models
  p(t_1,\dots,t_n) \}$.
\end{lemma}

\begin{proof}
  It is well known that without nesting each Datalog query $Q$
  corresponds to a such formula $\phi$. In fact, each literal $L$ is
  translated into a first order literal $L*$. Each rule
  $L_0 \leftarrow L_1,\dots,L_n$ is translated as a first order
  formula
  $\forall Z_1\dots\forall Z_n (L_1 \land \dots \land L_n \rightarrow
  L_0)$ where $Z_1,\dots,Z_n$ are the non free variables of the rule.
  And the whole program is translated as the conjunction of the
  formulas corresponding to each rule.

  The proof is done by translating each the atom that we added to
  Datalog into first order formulas, i.e., Nested queries. Suffices
  giving a recursive translation for them. We define it as follows.
  The formula of a nested query $Q=(p(X_1,\dots,X_n),P)$ is the first
  order formula
  $\exists Y_1 \dots \exists Y_n(p'(Y_1,\dots,Y_n) \land
  X_{j_1}=Y_{j_1} \land \dots \land X_{j_m}=Y_{j_m} \land \phi_P)$,
  where $\{j_1,\dots,j_m\}\subseteq\{1,\dots,n\}$ and $\phi_P$ is the
  formula of $P$ after renaming each predicate $q$ in $P$ with a fresh
  predicate $q'$ (this ensures that the evaluation is isolated from
  the outside of the nested query), and $X_{j_1},\dots,X_{j_m}$ are
  the variables that occur positively in the rule where $Q$ is nested.
  The equalities $X_{j_1}=Y_{j_1} \land \dots \land X_{j_m}=Y_{j_m}$
  model the logical substitution. Also, free variables $X$ in rules of
  $P$, are not added in $\forall X$ quantifiers of rules of $P$, so
  they are syntactically substituted.  Improper substitution of a
  variable $X$ in a rule $E$ can be simulated by removing the
  $\forall X$ in the formula of $R$.  \qed
\end{proof}

A corollary of this lemma is that nesting does not add expressive
power to Datalog. Hence, if we chose using one form of substitution
but not the other (e.g., using logical substitution but not
syntactical), we will have the same expressive power than if we had
chosen the contrary, or no substitution.

\section{Modal Datalog}

Modal Datalog is a version of Datalog where each rule is 
labeled with a modal logical operator. In what follows we 
develop it.

Most real world information includes incomplete data.  The main
device to codify incompleteness has been the null values.  A null,
denoted $\bot$, represents either that the value is missing or non
applicable. A relation containing null values is said
\textit{incomplete}, while one without them is said \textit{complete}.
The semantics of an incomplete relation is the set of all possible
complete relations resulting from replacing consequently each null by a
constant \textit{or} a symbol $\top$, denoting a non-applicable value.
We follow the semantics of $\bot$ and $\top$ by
Lerat and Lipski~\cite{DBLP:journals/tcs/LeratL86}.  However, in this
paper we only consider $\bot$ nulls, because the evaluation process of
Modal Datalog never generates non applicable values if they are not
present in the database.

The answer of a query in a complete database $D$ is
characterized by the set $\{\mu \mid D,\mu \models \phi\}$ where $\phi$
is a first order formula whose free variables are instantiated by $\mu$. 
On the other hand, an incomplete database is
interpreted as a set of complete databases $\mathcal{X}$, the 
possible worlds. In this context, we can give
modal characterizations to an answer $\mu$:  $\mu$ 
is \textit{sure}, denoted
$\mathcal{X},\mu\models\Box\phi$, if $D,\mu\models\phi$ for all
$D\in\mathcal{X}$. Similarly, $\mu$ is a \textit{maybe} answer,
denoted $\mathcal{X},\mu\models\Diamond\phi$, if there exists a
database $D\in\mathcal{X}$ such that $D,\mu\models\phi$.

{\em Modal Datalog} essentially introduces a mode for each rule $R$ in
a Datalog program. If the mode of $R$ is $\Box$, then $R$ is said to be sure
and infers facts that are valid for all instances of the null values
occurring in $R$ in the current database. Otherwise, if the mode of
$R$ is $\Diamond$, then $R$ is said to be a maybe rule and infers
facts that are valid for at least one instance of the null values.

\begin{definition}[Modal Datalog Syntax]
  A Modal Datalog Program is a set of rules of the form
  $\circ(p(X_1,\dots,X_n) \leftarrow B)$ where $\circ$ is either
  $\Box$ or $\Diamond$, where the symbol $\bot$ can occur in the body
  of the rule in the same places than terms.
     
  A modal Datalog query is one built with Modal Datalog programs.
\end{definition}

Next we present how modal predicates are derived from sets of facts
and substitutions. We write $S,\theta \models \circ L$ if $L$ is
derived from $S$ and $\theta$ with the label $\circ$ sure or maybe.
We say that a Modal Datalog predicate formula $L_1$ is less
informative than another $L_2$, denoted $L_1 \leq L_2$, if every
instance of $L_2$ is an instance of $L_1$.  Let $L$ be a literal, $S$
be a set of facts, and $\theta$ be a substitution. Then:
\begin{itemize}
\item $S,\theta \models \Box L$ if one of the following
  conditions holds:
  \begin{itemize}
  \item $L$ is a positive predicate formula and $\theta(L) \in S$.
  \item $L$ is a negative predicate formula $\neg q(t_1,\dots,t_m)$ and
    there does not exist a fact $q(a_1,\dots,a_m)$ in $S$ such that
    $\theta \models \Diamond(t_j=a_j)$, for $1 \leq j \leq m$.
  \item $L$ is $\define(X=t)$ and $\theta(X)$ is $t$.
  \item $L$ is $\filter(t_1=t_2)$ or $\filter(t_1 \neq t_2)$ and
    $\theta\models\Box(L)$
  \item $L$ is $(L',P')$ and $(L',\theta(P'))$ has at least one answer.
  \item $L$ is $\neg(L',P')$ and $(L',\theta(P'))$ has no answers.
  \end{itemize}
\item $S,\theta \models \Diamond L$ if one of the
  following conditions holds:
  \begin{itemize}
  \item $L$ is a positive predicate and there exists a fact
    $F \in S$ such that $F\leq\theta(L)$.
  \item $L$ is a negative predicate formula $\neg q(t_1,\dots,t_m)$
    and there does not exist a fact $q(a_1,\dots,a_m)$ in $S$ such that
    $\theta \models \Box(t_j=a_j)$, for $1 \leq j \leq m$.
  \item $L$ is $\define(X=t)$ and $\theta(X)$ is $t$.
  \item $L$ is $\filter(t_1=t_2)$ or $\filter(t_1 \neq t_2)$ and
    $\theta\models\Diamond(t_1 \neq t_2)$.
  \item $L$ is $(L',P')$ and $(L',\theta(P'))$ has at least one answer.
  \item $L$ is $\neg(L',P')$ and $(L',\theta(P'))$ has no answers.
  \end{itemize}
\end{itemize}

\begin{definition}[Semantics of Modal Datalog]
  \label{def:mrd-inference}
  Given a Modal Datalog query $(L,P)$, a database $E$, and a
  set of already inferred facts $S$, a fact $F$ is inferred
  from a rule $\circ (H \leftarrow B)$ in $P$
  ($\circ$ is either $\Box$ or $\Diamond$)
  if and only if there exists a substitution $\theta$ such that
  $S,\theta \models \circ L$ for all literals $L$
  in $B$, $\theta$ is the less informative substitution $\theta'$
  such that $S,\theta' \models \circ L'$ for all positive predicate
  formulas $L'$ in $B$, and $\theta(H)=F$.
\end{definition}

\begin{example}
Let $E=\{r(a),s(\bot),t(\bot)\}$ be a set of facts and $Q=(p(X),P)$ be the
query where the program $P$ has the rules
$R_1:\Box(p(X) \leftarrow q(X),r(X))$ and
$R_2:\Diamond(q(X) \leftarrow s(X),t(X))$.
Then, let us evaluate $\answer_E(Q)$.
The literal $r(X)$ in $R_1$ is only true with the substitution $\{X/a\}$,
because $r(a)$ is the only available fact. To infer $p(a)$ we need
first to infer $q(a)$ from rule $R_2$. Despite rule $R_2$ finds possible
answers, inferring $q(a)$ is not possible
because $R_2$ infers the less informative fact, i.e., $q(\bot)$.
Hence, $\answer_E(Q)$ is empty.
\end{example}

Now we are ready to define the promised notion of substitution in
Modal Datalog. As before, we will make a distinction between syntactical
substitution and logical substitution.

In the case of {\em syntactical substitution}, the approach is the same. 
We add a literal $\define(X=\theta(X))$ to the rule where $X$ occurs free.
As this literal is syntactic sugar for introducing a predicate
formula $u(X)$, then $\define(X=\theta(X))$ is indeed a positive predicate
formula, so $X$ occurs positively in it
(see Def.~\ref{def:mrd-inference}). Then, $X$ is not free anymore, and
takes the value of $\theta(X)$, that can be a null or a constant.

The {\em logical substitution} is more subtle.
 As we saw, in logical substitutions the value
of a variable $X$ comes from more than one source. For example, if one
simply adds the literal $\define(X=\theta(X))$ to the body of the rule
$\circ(p(X)\leftarrow q(X))$ ($\circ$ is either $\Box$ or $\Diamond$),
then the value of $X$ will be provided by two sources, namely
$\theta(X)$ and the literal $q(X)$. The problem is what to do
if in one source $X$ is null and in the other a constant.
We solve this problem by splitting the rule in two, so one preserves
the mode $\circ$ of the original rule and in the other we use the
mode $\Diamond$ to merge the values coming from both sources.

The logical substitution of $X$ in a rule
$R=\circ(p(X_1,\dots,X_n) \leftarrow B)$
is the replacement of $R$ by the rules
$\Diamond(p(X_1,\dots,X_n) \leftarrow
u(X_1,\dots,X_n),\define(X=\theta(X)))$ and
$\circ(u(X_1,\dots,X_n) \leftarrow B)$, where $u$ is a fresh
predicate.
 The mode $\Diamond$ in the first rule checks the compatibility
and follows the SPARQL design, where a null value
is considered compatible with a constant. 

We showed that in Nested Datalog logical substitutions can be moved down in
the dependency graph of a program without changing its results
(Lemma~\ref{lemma:equivalence-top-bottom}). The following Lemma
states that this feature does not hold in Modal Datalog.

\begin{lemma}
  \label{lemma:no-equivalence-top-bottom}
  Given a substitution $\theta$ and a Modal Datalog program $P$,
  moving down literals of the form $\define(Y=\theta(X))$
  in the dependency graph of $P$ could change the semantics of $P$.
\end{lemma}

\begin{proof}
  It suffices to show an example that witnesses this.
  Consider the query $Q=(p(X,Y),P)$
  where $P$ is the unique rule
  $\Box(p(X,Y) \leftarrow r(X), s(Y), \filter(X=Y)$. 
  \setlength\tabcolsep{0.3em}
  \begin{center}
	\begin{tabular}{|c|c|}
		\toprule
		Program A & Program B \\ \midrule
		{\small
			\(
			\begin{array}{rl}
			\Diamond(p(X,Y) &\leftarrow u(X,Y), \define(X=\theta(X)))\\
			\Box(u(X,Y) &\leftarrow r(X), s(X), \filter(X=Y))
			\end{array}
			\)
		}
		&
		{\small
			\(
			\begin{array}{rl}
			\Box(p(X,Y) &\leftarrow u(X), s(Y), \filter(X=Y)) \\
			\Diamond(u(X) &\leftarrow r(X), \define(X=\theta(X)))\\
			\end{array}
			\)
		}\\
		\bottomrule
	\end{tabular}
  \end{center}
  In a database containing the facts $r(\bot)$ and $s(a)$,
  and $\theta=\{X/a\}$, program A  will have no
  solutions and program B will have the solution $p(a,\bot)$. \qed
\end{proof}

The example in the previous proof shows:

\begin{corollary}
In Modal Datalog
bottom-up and top-down evaluations do not behave in the same manner.
\end{corollary}

\section{Substitution in FILTER EXISTS expressions}

By using the machinery of Modal Datalog, we will model the evaluation of FILTER EXISTS expressions in SPARQL. 
Our objective is to provide a framework for safe semantics for correlated subqueries.


\subsection{SPARQL codified as Modal Datalog}

First we show that SPARQL can be coded in Modal Datalog using a translation
inspired by~\cite{DBLP:conf/semweb/AnglesG08}, that is, via relational algebra. The syntax and semantics of the SPARQL fragment studied here are defined using Relational Algebra with set semantics (as is done in \cite{cyganiak2005relational} and \cite{angles2016}).

We write $r(R)$ to denote a relation $r$ with schema (attributes) $R$, that is,
a set of mappings $\mu$ from $R$ to constants or null values.
We extend standard relational algebra to handle null values by using 
modal evaluation when needed:
\begin{align*}
\pi_{X_1,\dots,X_n}(r) &= \{ \mu[X_1,\dots,X_n] \mid \mu \in r \}, \\
\rho_{X/Y}(r) &= \{ \mu[X/Y] \mid \mu \in r \}, \\
r \cup s &= (r \times \bot^{S\setminus R}) \cup (s \times \bot^{R\setminus S}),\\
\sigma_{\Box\phi}(r) &= \{ \mu \in r \mid \mu \models \Box \phi \},\\
r -_{\Box} s &= \{ \mu \in r \mid \nexists \mu' \in s \forall X \in R\cap S : \Diamond(r[X] = s[X]) \},\\
r \Join_\Diamond s &= \{ \mu_1 {^\frown} \mu_2 \mid \mu_1 \in r, \mu_2 \in s,
\text{ and } \forall X \in R \cap S : \Diamond(r[X] = s[X]) \},
\end{align*}
where $r(R),s(S)$ are two relations; $\mu[T]$ is the truncation of the
tuple $\mu$ to the set of attributes $T$; $\mu[X/Y]$ is the renaming
of the attribute $X$ by $Y$; $\mu_1 {^\frown} \mu_2$ is the
concatenation of tuples, where $X$ takes the most informative value
for each common attribute $X$, or the available value if $X$ is not
common; and $\bot^T$ is the relation $t(T)$ with a single tuple filled
with null values.

Here we study the fragment of SPARQL composed by the operators
$\sigma_{\Box\phi}$, $\pi_{X_1,\dots,X_n}$, $\rho_{X/Y}$, $\cup$,
$\Join_\Diamond$, $-_\Box$. Note that we have selected one mode
  for each modal operator.
  The difference $-_\Box$ corresponds to the operator MINUS
in~\cite{DBLP:journals/tods/PerezAG09}, referred as DIFF
in~\cite{DBLP:conf/amw/AnglesG16}.\footnote{The standard MINUS is
  slightly different in the case when the subtracting mapping has no
  attributes, but both can be mutually simulated (see
  \cite{DBLP:conf/kr/KontchakovK16,DBLP:conf/amw/AnglesG16}).}  The
fragment where the operands of $\cup$ have the same attributes
precludes the emergence of nulls when evaluating databases without
nulls, thus coincides with Relational Algebra.
Otherwise, the extended algebra is required.  In this context, the
operator OPTIONAL, denoted here as $\LeftOuterJoin$, is defined in
\cite{DBLP:journals/tods/PerezAG09} as
$R \LeftOuterJoin S = (R \Join_\Diamond S) \cup (R -_\Box
S)$.\footnote{This definition of $\LeftOuterJoin$ is slightly
  different with such stated by the standard.  However, it is well
  known that the standard SPARQL operators are definable in the
  algebra presented here (e.g., see
  \cite{DBLP:conf/kr/KontchakovK16,DBLP:conf/amw/AnglesG16}).}

\begin{definition}[From algebra to Modal Datalog]
\label{def:algebra-to-mrd}
Given two relations $r(R)$ and $s(S)$, the Modal Datalog rules for
each algebraic operator are defined as follows:
\begin{align*}
\sigma_{\Box\phi}(r) &:\quad
\begin{aligned}[t]
&\Box(p(R) \leftarrow r(R),\filter(\phi))
\end{aligned}
\\
\pi_T(r) &:\quad
\begin{aligned}[t]
\Diamond(p(T) \leftarrow r(R))
\end{aligned}
\\
\rho_{X/Y}(r) &:\quad
\begin{aligned}[t]
\Diamond(p((R \setminus \{X\}) \cup \{Y\}) \leftarrow r(R),\filter(X=Y),\define(Y=\bot))
\end{aligned}
\\
r \cup s &:\quad
\begin{aligned}[t]
&\Diamond(p(R \cup S) \leftarrow r(R),\define(X_1=\bot),\dots,\define(X_n=\bot))  \text{ and}\\
&\Diamond(p(R \cup S) \leftarrow s(S),\define(Y_1=\bot),\dots,\define(Y_m=\bot)), \\
&\text{where }\{X_1,\dots,X_n\}=S\setminus R
\text{ and }\{Y_1,\dots,Y_n\}=R\setminus S
\end{aligned}
\\
r \Join_\Diamond s &:\quad
\begin{aligned}[t]
\Diamond(p(R\cup S) \leftarrow r(R),s(S))
\end{aligned}
\\
r -_\Box s &:\quad
\begin{aligned}[t]
\Box(p(R) \leftarrow r(R),\neg q(R\cap S)) \text{ and }
\Box(q(R\cap S) \leftarrow s(S))
\end{aligned}
\end{align*}
\end{definition}

The translation of each algebraic operator into a set of Modal Datalog
rules by Def.~\ref{def:algebra-to-mrd} allows translating 
algebraic expressions into Modal Datalog queries.

\begin{example}
  Given the relations $r(X,Y)$, $s(X,Y)$ and $t(X,Z)$, the
  expression $(r \Join_\Diamond (s -_\Box t))$ is translated as the
  query $(p(X,Y),P)$ where the program $P$ has the
  rules $\Diamond(p(X,Y) \leftarrow r(X,Y), q(X,Y))$,
  $\Box(q(X,Y) \leftarrow s(X,Y), \neg u(X))$, and
  $\Box(u(X) \leftarrow t(X,Z))$.
\end{example}

\begin{lemma}
  Let $Q$ be an algebra expression and $Q^*$ be the Modal Datalog
  query obtained from $Q$ according to Def.~\ref{def:algebra-to-mrd}.
  For every database $E$, it holds that $Q$ and $Q^*$ are equivalent,
  i.e. their evaluation results in the same answers.\footnote{Note
    that answers of $Q$ have the form $(a,b)$ while answers of $Q^*$
    have the form $p(a,b)$. In this lemma we assume that these answers
    are the same as they have the same components.}
\end{lemma}

\subsection{Modal Datalog semantics of FILTER EXISTS}
\label{sec:engines}

Hitherto, we have a semantics for a SPARQL fragment and a translation
to Modal Datalog, except for expressions of the forms $\sigma_Q(P)$
and $\sigma_{\neg Q}(P)$, where $P$ and $Q$ are called respectively the
outer and the inner patterns.\footnote{In the standard
  syntax these operators correspond to ($P$ FILTER EXISTS $Q$) and
  ($P$ FILTER NOT EXISTS $Q$), respectively.}

The philosophy of these operators is the following.  Given a relation
$r$ and an algebraic expression $Q$, we have that $\sigma_Q(r)$ and
$\sigma_{\neg Q}(r)$ return the set of tuples $\mu$ where $\mu(Q)$ has
a solution or no solutions, respectively.  According to the
SPARQL specification, $\mu(Q)$ is the result of replacing in $Q$ each
variable $X$ in the domain of $\mu$ by $\mu(X)$.  As we indicated in
the introduction, this definition is ambiguous and contradictory with
other parts of the specification, and (as expected) systems have
different interpretations for it.

We will unveil this problem by showing how the
definition of $\mu(Q)$ is viewed in Modal Datalog where the
notion of substitution shows up in a clean logical manner.

\begin{definition}[Filter Exists]
Given a relation $r$ and a SPARQL query $Q$, 
the expressions $\sigma_Q(r)$ and $\sigma_{\neg Q}(r)$ 
are translated to the Modal Datalog rules 
$\Box(p(T_r) \leftarrow r(T_r), (L,P))$ 
and
$\Box(p(T_r) \leftarrow r(T_r), \neg(L,P))$
respectively.
\end{definition}

Now we are ready to enumerate three sources of discrepancy in the interpretation
of a rule $R=\circ (H \leftarrow B,(L,P))$:\footnote{Due to space 
limitations, here we consider only the positive case $\sigma_{Q}(r)$.
  It is not difficult to extend the results for the negative case
  $\sigma_{\neg Q}(r)$.}
\begin{enumerate}
\item \textit{Free variables.}
  A free variable $X$ in $P$ is in some cases assumed
  uncorrelated. The lack of correlation of a variable is simulated
  by replacing $R$ by the rule
  $\circ (H \leftarrow B,(L,P'),\define(X'=\bot))$,
  where $X'$ is a fresh variable and $P'$ is the result of
  replacing each occurrence of $X$ in $P$ by $X'$.
\item \textit{Improper substitution.}
  Some engines implement improper substitution.
\item \textit{Substitution level.}
A variable $X$ in the goal of the nested query can be substituted
logically in different places, ranging from top-down to bottom-up
logical substitution. These substitutions are not equivalent
(see Lemma~\ref{lemma:no-equivalence-top-bottom}).
\end{enumerate}

\newcommand{\persons}{\operatorname{persons}}

Next we present example queries to show the ways the systems address the above sources of discrepancy.
We will use the dataset presented in the introduction.
We assume that $\persons(X)$ is a 
pattern giving all persons $X$ in the database, and $m_1(X,Y)$ and
$m_2(X,Y)$ are 
patterns returning respectively the corporate
and personal emails $Y$ of a person $X$. 

\paragraph{The first two discrepancies.}
Consider the queries
$\sigma_{m_2(X,Z)}(T)$,
$\sigma_{\pi_Z(m_2(X,Z))}(T)$ and
$\sigma_{\sigma_{Y=X}(m_2(Y,Z))}(T)$,
where $T$ is $m_1(X,\text{\tt*.com})$.
These queries can be simulated with the same Modal Datalog query
having goal $p(X)$ and a single rule
$\Box(p(X)\leftarrow m_1(X,b),Q)$, where $Q$ is either
$Q_1 = (q(X,Z),\Box(q(X,Z) \leftarrow m_2(X,Z)))$,
$Q_2 = (q(Z),\Box(q(Z) \leftarrow m_2(X,Z)))$, or
$Q_3 = (q(Y,Z),\Box(q(Y,Z) \leftarrow m_2(Y,Z),\filter(Y=X)))$.
Now, we have the following cases:
\vspace{-0.5em}
\begin{itemize}
\item Case $Q=Q_1$: It has a unique alternative which is substituting
  $X$ logically, so returning only person 1. All systems agree with
  this answer.
\item Case $Q=Q_2$: It has two interpretations, namely allowing or not
  allowing improper substitution. In the first, answers include only
  person 1. In the second answers are persons 1 and 3.  Blazegraph and
  Fuseki agree with the first interpretation, while rdf4j and Virtuoso
  with the second.
\item Case $Q=Q_3$: It has two interpretations, depending whether $X$
  is assumed correlated or not.  In the first interpretation, person 1
  is the unique answer.  In the second interpretation, there are no
  answers because $X$ is $\bot$ when evaluating $\filter(Y=X)$.
  Blazegraph, rdf4j and Virtuoso agree with the first interpretation,
  while Fuseki with the second.
\end{itemize}

\begin{figure}[t]
	\begin{tikzpicture}[level 4/.style={sibling distance=0.2cm},level 1/.style={sibling distance=1.3cm},level distance=3.7em]
	\node at (0.8,0) {$(p(X,Y),P)$}
	child {
		node { $\Box\frac{q(X,Y),\filter(Y=\,\text{\tt*.com})}{p(X,Y)}$}
		child { node (E) {$\Box\frac{r(X,Y)}{q(X,Y)}$} }
		child [missing]
		child [missing]
		child [missing]
		child [missing]
		child {
			node (G) {$\Box\frac{s(X),w(Y)}{q(X,Y)}$}
			child {
				node (C) {$\Box\frac{u(X),\neg t(X)}{s(X)}$}
				child [missing]
				child [missing]
				child {
					node (T) {$\Box\frac{v(X,q,Y)}{t(X)}$}
					child { node(A) {$\Box\frac{m_2(X,Y)}{v(X,Y)}$}}
				}
			}
			child [missing]
		}
	}
	;
	\node (B) [left=0.1cm of A] {$\Box\frac{\persons(X)}{u(X)}$};
	\node (D) at (T -| E) {$\Diamond\frac{u(X),v(X,Y)}{r(X,Y)}$}
	child { node {$\Box\frac{\persons(X)}{u(X)}$}}
	child [missing]
	child { node {$\Box\frac{m_2(X,Y)}{v(X,Y)}$}};
	\node (F) [right=0.1cm of A] {$\Box\frac{\define(Y=\bot)}{w(Y)}$};
	\draw[dotted] (-6,-0.7) -- (6,-0.7); \node at (-5.5,-0.5) {Level 1};
	\draw[dotted] (-6,-1.9) -- (6,-1.9); \node at (-5.5,-1.7) {Level 2};
	\draw[dotted] (-6,-5.5) -- (6,-5.5); \node at (-5.5,-5.3) {Level 3};
	\draw (C) -- (B);
	\draw (E) -- (D);
	\draw (G) -- (F);
	\node at (1,-0.5) {\scriptsize 1};
	\node at (-1.2,-1.7) {\scriptsize 2};
	\node at (2.9,-1.7) {\scriptsize 3};
	\node at (-3.3,-5.3) {\scriptsize 4};
	\node at (-1.6,-5.3) {\scriptsize 5};
	\node at (2,-5.3) {\scriptsize 6};
	\node at (3.8,-5.3) {\scriptsize 7};
	\node at (5.3,-5.3) {\scriptsize 8};
	\end{tikzpicture}
	\caption{Dependency graph of the Nested Datalog query for the inner
		pattern $Q$.}
	\label{fig:dependency-tree}
\end{figure}
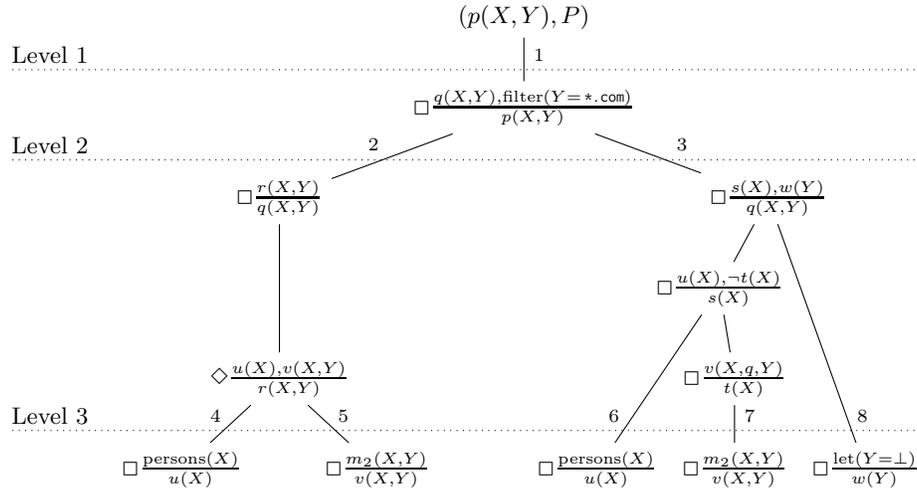

\vspace{-1.5em}
\paragraph{The last discrepancy.}
To check the application of logical substitution, consider the query
\(\sigma_Q(\persons(X)\LeftOuterJoin(m_2(X,Y)))\) where $Q$ is
$\sigma_{Y=\text{\tt*.com}}(\persons(X)\,\LeftOuterJoin\,m_2(X,Y))$,
corresponding to the one presented in the introduction.
Fig.~\ref{fig:dependency-tree} depicts the dependency graph of
$Q$.  
Some edges are labeled to refer alternative places where a
substitution $\theta$ can be applied. 
For instance, the top-down approach consists in inserting the rule
$\Diamond(p(X,Y) \leftarrow
p'(X,Y),\define(X=\theta(X),\define(Y=\theta(Y)))$ in the edge 1, and
replacing $p$ in the rule below by $p'$, where $p'$ is a fresh
predicate.  
Hence, only person 5 succeeds. No system agrees with this evaluation.

If substitutions are done in Level 2, then persons 3 and 5 succeed.
Only rdf4j agrees with this interpretation.

Blazegraph and Fuseki agree with substitution in Level 3, i.e., just
after variables instantiation. 
Strictly, this approach is not logical as it has some improper substitutions. 
Indeed, the logical chain of variable $Y$ below the edge 7 
does not start in the goal of the nested query. 
Here persons 1, 3 and 5 succeed.

Virtuoso, Blazegraph and Fuseki follow a similar approach 
but they also append the literals $\define(X=\theta(X))$ and
$\define(Y=\theta(Y))$ to the rule below edge 1. 
Hence, person 5 is discarded by the filter on this rule.

\section{Conclusions}

This work shows that the notion of substitution continues to haunt
researchers and developers. We showed that it is at the core of the
subtle problems that the specification and implementations of SPARQL
face regarding subqueries of the form FILTER EXISTS.

We think the lasting contribution of this paper is the finding of
several types of substitution in the presence of nested expression and
incomplete information, that are playing some roles and seems that had
passed unnoticed until now.

Although we consciously did not advance any proposal to fix the
problems of substitution in nested expressions in SPARQL, the paper
leaves a chart with the possible avenues to solve them. We think that
the Working Groups of the W3C have the last word on this issue.

\bibliographystyle{plain}

\begin{thebibliography}{99}
\bibitem{DBLP:conf/semweb/AnglesG08} R. Angles and C. Gutierrez. The
  Expressive Power of SPARQL. In Proc. of the International Semantic
  Web Conference (ISWC), volume 5318 of LNCS, pages 114–129. Springer,
  2008.
\bibitem{angles2016} R. Angles and C. Gutierrez. The Multiset
  Semantics of SPARQL Patterns. In Proc.  of the International
  Semantic Web Conference (ISWC), volume 9981 of LNCS, pages
  20–36. Springer, 2016.
\bibitem{DBLP:conf/amw/AnglesG16} R. Angles and C. Gutierrez. Negation
  in SPARQL. In Proc. of Alberto Mendelzon International Workshop on
  Foundations of Data Management (AMW), volume 1644 of CEUR Workshop
  Proceedings, 2016.
\bibitem{cardone2006history} F. Cardone and J. R. Hindley. History of
  lambda-calculus and combinatory logic.  Handbook of the History of
  Logic, 5:723–817, 2006.
\bibitem{church1996introduction} A. Church. Introduction to
  mathematical logic, volume 13. Princeton University Press, 1996.
\bibitem{cyganiak2005relational} R. Cyganiak. A Relational Algebra for
  SPARQL. Digital Media Systems Laboratory HP Laboratories
  Bristol. HPL-2005-170, page 35, 2005.
\bibitem{DBLP:conf/iclp/DongL92} F. Dong and
  L. V. S. Lakshmanan. Deductive databases with incomplete infor-
  mation. In Proc. of the Joint International Conference and Symposium
  on Logic Programming (JICSLP), pages 303–317. MIT Press, 1992.
\bibitem{DBLP:journals/tcs/DongL94} F. Dong and
  L. V. S. Lakshmanan. Intuitionistic Interpretation of Deductive
  Databases with Incomplete Information. Theor. Comput. Sci.,
  133(2):267–306, 1994.
\bibitem{DBLP:journals/jlp/GiordanoM94} L. Giordano and
  A. Martelli. Structuring logic programs: A modal approach. J.
  Log. Program., 21(2):59–94, 1994.
\bibitem{sparql11} S. Harris and A. Seaborne. SPARQL 1.1 Query
  Language - W3C Recommendation, March 21 2013.
\bibitem{DBLP:journals/corr/HernandezGA16} D. Hernandez, C. Gutierrez,
  and R. Angles. Correlation and substitution in SPARQL. CoRR,
  abs/1606.01441, 2016.
\bibitem{DBLP:journals/tois/KongC95} Q. Kong and G. Chen. On Deductive
  Database with Incomplete Information. ACM Trans. Inf. Syst.,
  13(3):355–369, 1995.
\bibitem{DBLP:conf/kr/KontchakovK16} R. Kontchakov and
  E. V. Kostylev. On expressibility of non-monotone operators in
  SPARQL. In Proc. of the International Conference on Principles of
  Knowledge Representation and Reasoning (KR), pages 369–379. AAAI
  Press, 2016.
\bibitem{DBLP:journals/tcs/LeratL86} N. Lerat and W. Lipski
  Jr. Nonapplicable Nulls. Theor. Comput. Sci., 46(3):67–82, 1986.
\bibitem{DBLP:books/daglib/0097689} M. Levene and G. Loizou. A guided
  tour of relational databases and beyond.  Springer, 1999.
\bibitem{DBLP:conf/semweb/Patel-Schneider16} P. F. Patel-Schneider and
  D. Martin. Existstential aspects of SPARQL. In Proc. of the
  International Semantic Web Conference (ISWC), Posters \&
  Demonstrations Track, 2016.
\bibitem{DBLP:journals/tods/PerezAG09} J. Perez, M. Arenas, and
  C. Gutierrez. Semantics and Complexity of SPARQL.  ACM
  Trans. Database Syst., 34(3):16:1–16:45, 2009.
\end{thebibliography}

\newpage

\appendix
\section{Moving logical substitutions to lower levels}

In this appendix we will describe how a substitution literal
$\define(X=\theta(Y))$ can be moved to lower levels of the
dependency graph without changing the semantics of a Nested
Datalog query. Because Nested Datalog programs $P$ are non
recursive, the structure of the dependency graph of $P$ is a
tree, so substitution literals can be moved down repeatedly
until the leaf predicates are reached.

We need another structure.
Given a program $P$, then we call the tree of $P$ to the rooted
tree whose nodes are labeled with the rules of $P$ and where a
node labeled rule $R_1$ is the child of a node labeled by
a rule $R_2$ if the head of $R_1$ is in the body of $R_2$.

Now, we will describe how to move down an literal $\define(Y=\theta(X))$
in a rule $R$ of a program $P$ where this literal occurs. Without
loss of generality, consider that $R$ is the following rule:
\[H \leftarrow L_1,\dots,L_n, \define(Y=\theta(X)),L'_1,\dots,L'_m,\]
where $L_1,\dots,L_n$ are positive predicates formulas with
intensional predicates where $Y$ occurs and $L'_1,\dots,L'_m$
are the rest of the literals in the body of $R$.

For $1 \leq i \leq n$, let $p(X_1,\dots,X_n)$ be $L_i$ and
$J_i$ be the set of positions $j$ of $L_i$ such that $X_j=Y$
(i.e., are the same variable). Then, rename $p$ in $L_i$ with
a fresh predicate name $p_i$ and do the following for each rule
$R'$ in $P$ of the form $p(Z_1,\dots,Z_n) \leftarrow B$:
\begin{enumerate}
\item Let $T$ be a copy of the thread whose root is $R'$ in the
  tree of $P$, and $T'$ the result of renaming consequently   
  all intentional predicates occurring in $T$ by a fresh predicate,
  except $p$ that is renamed as $p_i$.
\item Append the literal $\define(Z_j=\theta(X))$ to the
  body of the root of $T'$, for each position $j$ in $J_i$.
\item Copy the rules of $T'$ into $P$.
\end{enumerate}

\begin{example}
Consider the following program:
\begin{align*}
p(X,Y) &\leftarrow q(X,Y), q(Y,X), \define(X=\theta(X))\\
q(X,Y) &\leftarrow r(X,Y)\\
r(X,Y) &\leftarrow s(X,Y)
\end{align*}
where $s$ is an EDB-predicate.
Then, moving the literal $\define(X=\theta(X))$
one level below results in the following program:
\begin{align*}
p(X,Y) &\leftarrow q(X,Y), q(Y,X)\\
q_1(X,Y) &\leftarrow r_1(X,Y), \define(X=\theta(X))\\
r_1(X,Y) &\leftarrow s(X,Y)\\
q_2(X,Y) &\leftarrow r_2(X,Y), \define(Y=\theta(X))\\
r_2(X,Y) &\leftarrow s(X,Y)
\end{align*}
\end{example}

\begin{definition}[Bottom-up logical substitution]
  Let $\theta$ be a substitution and $Q=(L,P)$ be a Nested  Datalog query.
  Then the bottom-up substitution $\theta(Q)$ is the query
  resulting after applying the top-down substitution of $\theta$ in $Q$
  and then moving down the added literals until they reached
  the leafs of the tree of $P$.
\end{definition}

Is not difficult to see that the process of moving a literal
down once level results in an equivalent program. Hence, we got the
following result.

\begin{lemma}
  Let $Q$ be a Nested Datalog query. Let $\answer^\shortdownarrow$ and
  $\answer^\shortuparrow$ be the respective evaluation procedures
  using top-down and bottom up substitution.  Then, for every
  extensional database $E$ for $Q$ holds
  $\answer_E^\shortdownarrow(Q)=\answer_E^\shortuparrow(Q)$.
\end{lemma}

\section{Comments on ESWC 2018 review}

This document is an extended version of a paper submitted to ESWC
2018. In this appendix, we will present some answers to questions
formulated by the reviewers. Some of them could be useful to
understand the notions defined here. We are currently working to
improve the readability of the paper. When we achieve the desired
clarity, we will remove this appendix.

It is important to recall that the first motivation of our work is
that the evaluation of correlated (sub)queries differs among the
current SPARQL engines (we presented examples of this fact). Second,
the origin of such differences comes from an ambiguous semantics
defined by the SPARQL 1.1 specification. Third, there is no standard
foundation or reference to explain the notion of correlated queries
(even for well-known studied languages like SQL, relational algebra,
and Datalog). Therefore, our goal was to provide a logical framework
to study possible interpretations and semantics for correlated
queries, and to understand the implementations of actual engines.

The importance of a logical foundation for query languages was noticed
by Reiter. He presented a variety of arguments in favor of doing
so. One of these arguments was that an evaluation argument could be
proved to be sound and complete with respect to the logical semantics
of the data model and the query. In this work we follow the
motivations of Reiter when providing a logical framework to understand
the notion of correlation in nested queries.

\bigskip
\noindent
Answers to questions.
\medskip

Reviewer 3 is right: we missed the paper of Kaminski et al. (ACM
2017), and it should be incorporated into the related work. We
reviewed the WWW'2016 version, but it does not mention the issues of
substitution. On the other side, the ACM'2017 version differs from our
work as it proposes a concrete correlation method. In our case, we aim
to provide a logical framework to try to understand alternative
semantics. For example, according to their proposal, the result of the
example query in the introduction of our paper is $(5,-)$. This result
is different to all solutions provided by the studied engines.

Reviewer 1: Concerning the modes of the operators join and minus. Let
us explain the case of the join (the case of the minus is
similar). Given two relations $r$ and $s$, the SPARQL operation
$(r \Join s)$ corresponds to the set of concatenations of compatible
tuples $\mu$ in $r$ and $\nu$ in $s$. That is, $\mu$ and $\nu$ holds
that $\mu(X)=\nu(X)$, $\mu(X)$ is unbound or $\nu(X)$ is unbound, for
each variable $X$ (Here we represent unbound values with the symbol
$\bot$.) This is equivalent to say that there exists an instance of
the tuples $\mu$ and nu such that $\mu(X)=\nu(X)$ (that we wrote in
modal notation). A tuple $(a,b)$ is usually called an instance (or
possible world as we mentioned previously in the paper) of
$(a,\bot)$. The same applies to literals (reviewer 4 asked
this). Reviewer 3 is also confused with the ``magic'' behavior of
$\bot$ in the modal evaluation of equalities. The
$\Box(\filter(a=\bot))$ is false because there is a instance (e.g.,
$\filter(a=b)$) where the literal is not true. On the other hand,
$\Diamond(\filter(a=\bot))$ is true because there is a instance (e.g.,
$\filter(a=a)$) where the literal is true.

In several parts of the document we use the term ``built-in''. It seems
that this notion is not clear to some reviewers. The built-in si a
formal notion in Datalog, e.g. see the survey of Ceri, Gottlob, and
Tanca, ``What You Always Wanted to Known About Datalog (And Never Dared
to Ask).'' A built-in is a literal that occurs in the body of a rule and
is evaluated after instantiating all variables occurring positively in
that rule. In this work nested queries and equality formulas are
built-ins.

Reviewer 2: Concerning the analysis of the two semantics discussed in
the related work. The answer is yes. The proposed by Patel-Schneider
and Martin evaluates the inner pattern first, so it corresponds to the
top-down logical substitution. On the other hand, the proposal of
Seaborne applies substitutions in the leaves of the syntactical tree
of an expression, so it corresponds to bottom-up logical substitution.

Reviewer 3: Concerning the expression ``equality formulas require that
all variables be assigned before being evaluated.'' The reviewer is
right in that the explanation is poor. To explain better the notion
that we are introducing, consider the rule
$p(X,Y) \leftarrow q(X), X=Y$. According to Levene and Loizou (``A
guided tour of relational databases and beyond''), if the database
contains a fact q(a) then the rule infers the fact
$p(a,a)$. Intuitively, the equality $X=Y$ ``pass'' the value from $X$
to $Y$. We do not use such notion because it brings problems when $Y$
is bound in the outer query. With our built-in equality, the rule
$p(X, Y) \leftarrow q(X), \filter(X=Y)$ is unsafe because $Y$ does not
occur positively. In other words, $Y$ is not assigned.

Reviewer 3: Concerning to Stratified Datalog. When Datalog is extended
with negation, two approaches could be used, Stratified or
Inflationary (see the survey of Ceri, Gottlob and Tanca of
1989). Datalog programs with negation satisfy several minimal Herbrand
models. This entails difficulties in defining the semantics of Datalog
programs. Stratified Datalog permit us choosing a distinguished
minimal Herbrand model by approximating the CWA. To the best of our
knowledge, all translation from SQL and SPARQL to Datalog consider
stratified programs.

Reviewer 3: About the notion of ``logically connected.'' It was
intended as an informal intuition just to say that two variables are
not logically independent (that is, the values they get are related
somewhat). For instance, consider the rules $p(X) \leftarrow q(X,Y)$
and $q(Z,U) \leftarrow r(Z,V), s(U)$. Then the variables $X$ and $Z$ are
logically connected. To be more precise, the connection is not between
the variables, but between the position where variables occur in each
rule. We can rename the variables consequently without breaking the
logical connection.

The notions of ``logical chain'', ``logical substitution'' and
``improper substitution'' are related to the notion of ``logically
connected''. A logical chain is a sequence of connections between
variables. We can apply a logical substitution $X/a$ in a variable $Y$
occurring in a rule of a nested query if there is a logical chain from
an occurrence of $X$ in the goal of the nested query to the variable
$Y$ in such rule. Similarly, the substitution is improper when there
is no such connection between the variable substituted Y and an
occurrence of $X$ in the goal of the nested query. One of our
contributions is showing that some engines do an improper
substitution.

What is the problem with improper substitution? In Sec. 3 we argue
that improper substitution compromises the compositionality of
nesting. As we explain in this appendix, one of the reasons for a
logical foundation is to provide a logical interpretation for
queries. In the proof of Lemma 2, we extend the usual translation from
Datalog programs to first order formulas by providing a translation
for the new type of atom: nested queries. In an ideal setting, each
nested query can be translated into a first order query, independently
of the outer query where it is nested. Intuitively, compositionality
is achieved when the meaning of the nested query does not depend on
the query where it is nested.  For instance, let $Q$ be a query nested
in a rule $R$ of an outer query, $X$ be a variable that occurs in a
rule $R'$ in $Q$, and $X$ be not logically connected with the goal of
$Q$. Then, allowing improper substitution, we have two possible
translations for $Q$ as a first order formula. If $X$ occurs
positively in $R$, then we remove the universal quantifier $\forall X$
to the rule $R'$. Otherwise, we do not remove it.

From reviewer 3: concerning the meaning of $\define(X=a)$. We write
$\{p(X,Y) \leftarrow q(X),\define(Y=a)\}$ to denote the program
$\{p(X,Y) \leftarrow q(X),l(Y); l(a)\}$. Thus, the program has a rule and a
fact.

From reviewers 3 and 4: The notice that in the table of page 6, when
$\define(X = \theta(X))$ is moved down, then it changes to
$\define(Y = \theta(X))$. Reviewers suggest that it could be a
typo. However, this is correct. The variables on the left side of the
equality correspond to the variable that is connected with X, that may
change depending on the rule. Also $\theta(X)$ does not change because
it is a constant in the current substitution.

Reviewer 4 complains that the semantics of $S,\theta \models P$ was
not explained for a program $P$. We do not define
$S,\theta \models P$, because it is not necessary to define the
semantics of a program. It suffices to define what facts are inferred
by a rule. The procedure for doing that is standard in the Datalog
literature. First, a true value is defined for literals in the body of
a rule with respect to a substitution. We write $S,\theta \models L$
to denote that a literal $L$ is evaluated as true with respect to a
substitution $\theta$ and a set of facts $S$. Second, a fact
$\theta(H)$ is inferred in a rule if $H$ is the head and for each
literal $L$ in the body of the rule, $L$ is true with respect to
$\theta$. This process ends when no more facts can be inferred from
all the rules of the program. Thus, the semantics of the program is
defined by this inference process.

For first order formulas $\phi$ and $\psi$, the notation
$\phi \models \psi$ has the standard meaning. That is, $\psi$ is a
logical consequence of $\psi$.

Reviewer 4: concerning to how variables in a negated subquery can be
positively defined. Variables are defined positively in the scope of a
rule. As subqueries are built-ins, a subquery $Q$ does not matter for
counting the variables occurring positively in the rule where $Q$
occurs. Also, the variables that are defined positively in each rule
of the subquery $Q$ do not depend on if $Q$ is negated or not.

Reviewer 4: concerning to what ``evaluate $P$'' means. It means
computing all the facts that can be inferred from P in the current
database.

Reviewer 4: concerning to what ``compatible'' means. Compatibility is
an standard SPARQL notion. Two mappings mu and nu are said to be
compatible if for each variable $X$ in the domain of mu holds
$\mu(X)=\nu(X)$ or $X$ is unbound in $\nu$.

Reviewer 4: concerning what ``moving down literals of the form
$\define(Y=\theta(X))$'' mean. It essentially means if $Y$ is logically
connected with another variable $Z$ that occurs in a rule below in
dependency graph of the program, then an equivalent literal
$\define(Z=\theta(X))$ is appended to such rule, and then removed from
the original rule. This process is detailed in Appendix A.

Reviewer 4: concerning what ``if the logical substitution in a query
uses the definition of substitution for a rule.'' The answer is
yes. We first defined how to substitute a variable $X$ by a value $a$
in a rule. Then we defined how to substitute $X$ by $a$ in a
program. It is not simply substituting $X$ by $a$ in a rule where $X$
occurs because in Datalog the scope of variables are the rules where
they occur. Thus, we need to identify the variables that are logically
connected and then apply a substitution of them, using the previous
definition of substitution in a rule. We denote the substitution of
$\theta$ in a program $P$ as $\theta(P)$.

Reviewer 4: concerning to ``extensional database.'' This is a standard
concept of Datalog that can be revised in the Datalog survey of Ceri
et al. mentioned before. Essentially, an extensional database is the
set of facts where the query is evaluated, and the intentional
database is such that includes the inferred facts.

Reviewer 4: about how the corollary follows from Lemma 2. First order
queries and standard Datalog have the same expressive power (see the
survey referred above). Lemma 2 says that Nested Datalog also has the
same expressive power than first order queries. Then, we conclude that
nesting does not add expressive power to Datalog. This applies to all
the presented methods of substitution because all of them were checked
in the proof of Lemma 2.

Reviewer 4: concerning the word EXISTStential in the
bibliography. This word is correct. It was part of the title of the
cited paper. The authors create this new world to make emphasize of
the keyword EXISTS.

\end{document}